\documentclass[letterpaper, 11 pt, conference]{ieeeconf}
\IEEEoverridecommandlockouts 
\usepackage{amsmath,amssymb,latexsym}
\usepackage[dvips]{graphicx}
\usepackage{graphicx}
\usepackage{psfrag}
\usepackage{epsfig}
\usepackage{float}
\usepackage{array}
\usepackage{mathrsfs}
\usepackage[geometry]{ifsym}
\usepackage{color}
\usepackage{fancyhdr}
\newtheorem{thm}{\bf Theorem}

\newtheorem{assumption}{Assumption}
\newtheorem{remark}{Remark}

\title{\bf Semi-Active Control of the Sway Dynamics for Elevator Ropes}
\author{ Mouhacine Benosman
\thanks{Mouhacine Benosman (m{\_}benosman@ieee.org) is with Mitsubishi Electric Research Laboratories, 201 Broadway Street, Cambridge, MA 02139, USA.}
 }

\date{}

\begin{document}
\maketitle\thispagestyle{empty}
\begin{abstract}
In this work we study the problem of rope sway dynamics control
for elevator systems. We choose to actuate the system with a
semi-active damper mounted on the top of the elevator car. We
propose nonlinear controllers based on Lyapunov theory, to actuate
the semi-active damper and stabilize the rope sway dynamics. We
study the stability of the proposed controllers, and test their
performances on a numerical example.
\end{abstract}
\section{Introduction}
Modern elevators installed in high-rise buildings are required to
travel fast and ensure comfort and safety for the passengers.
Unfortunately, the dimensions of such high-rise buildings make
them more susceptible to the impact of bad weather conditions.
Indeed, when an external disturbances, like wind gust or
earthquake, hits a building it can lead to large rope sway
amplitude within the elevator shaft. Large amplitudes of rope sway
might lead to important damages to the equipments that are
installed in the elevator shaft and to the
 elevator shaft structure itself, without mentioning the
 potential danger caused for the elevator passengers. It is important  then to be able to control the elevator system and damp out these
 ropes oscillations. However,
 due to cost constraints, it is preferable to be able to do so,
 with minimum actuation capabilities. Many investigations have been
 dedicated to the problem of modelling and control of elevator ropes
 \cite{K011,KIT09,ZX03,ZT03,ZC06,FWP93,B14}. In \cite{ZT03}, a scaled model for
 high-rise, high-speed elevators was developed. The model was used
 to analyze the influence of the car motion profiles on the
 lateral vibrations of the elevator cables. The author in \cite{K011} proposed a nonlinear modal
 feedback to drive an actuator pulling on one end of the rope.
 The control performance was investigated by numerical tests. In \cite{FWP93},
 a simple model of a cable attached to an actuator at its free
 end was used to investigate the stiffening effect of the control
 force on the cable. An off-line energy analysis was used to tune an open-loop sinusoidal force
 applied to the cable. In \cite{ZC06}, the authors proposed a novel idea to dissipate the
 transversal energy of an elevator rope. The authors used a passive
 damper attached between the car and the rope. Numerical analysis
 of the transverse motion average energy was conducted to find the optimal
 value of the damper coefficient. In \cite{B14,B14_1,B14_2}, the present author studied the problem of an elevator system using a force actuator pulling on the ropes to add control tension to the ropes. The
 force actuator was controlled based on Lyapunov theory. Although, the sway damping performances obtained in \cite{B14,B14_1,B14_2} were satisfactory, we decided
 to investigate other forms of actuation and control. Indeed, in
 \cite{B14,B14_1,B14_2} we proposed a force control algorithm to modulate the
 ropes tensions, using an external force actuator, which was
 introduced at the bottom of the elevator shaft to pull on the
 ropes. However, retrofitting existing elevators with such a
 cumbersome actuator, can prove to be challenging and expensive.
 Instead, we propose to investigate here another actuation method,
 namely, using semi-active dampers mounted between the elevator
 car and the ropes. These actuators are less cumbersome than the
 typical force actuators and thus can be easier/cheaper to
 install.
This type of actuation has already been proposed in  \cite{ZC06}.
The difference between the present work and
 \cite{ZC06}, is that instead of using a static damper with constant damping coefficient tuned off-line, we use a
 semi-active damper, and use nonlinear control theory to design a
 feedback controller to compute online the desired time-varying damping
 coefficient that reduces the rope sway. We study the stability of the closed-loop dynamics, and show the
 performances of these controllers on a numerical example. One
 more noticeable difference with the work in \cite{B14}, is that
 due to the semi-active actuation, the model of the actuated system is
 quite different from the one presented in \cite{B14}. Indeed, in
 \cite{B14} the model of the actuated system exhibited terms where
 the control variable was multiplied by the sway position variable
 only. In this work, the actuated model, exhibits terms where the
 control variable, i.e. the semi-active damper coefficient, is
 multiplied both by the sway position and the sway velocity
 variables, furthermore, the control variable term appears in the right hand
 side of the dynamical equations of the system, i.e. as part of an
 external disturbance on the system (refer to Section
 \ref{section1}). These differences in the model, make the
 controller design and analysis more challenging than in
 \cite{B14}.
 \\\\The paper is organized as follows: In Section
 \ref{section1}, we recall the model of the system when actuated with a semi-active damper mounted between the ropes and the elevator car.
  Next, in Section \ref{section2}, we present the main results of this work, namely, the nonlinear
 Lyapunov-based semi-active damper controllers, together with their stability
 analysis. Section \ref{section3} is dedicated to some numerical
 results. Finally, we conclude the paper with a brief summary of the results in Section
 \ref{section4}.\\
Throughout the paper, $\mathbb{R}$, $\mathbb{R}_{+}$ denotes the
set of real, and the set of nonnegative real numbers,
respectively. For $x\in \mathbb{R}^{N}$ we define
 $|x|=\sqrt{x^{T}x}$, and we denote by $A_{ij},\;i=1,...,n,\;j=1,...,m$ the elements of the matrix $A$.
\section{Elevator Rope Modelling}\label{section1} In this section we first
recall the infinite dimension model, i.e partial differential
equation (PDE), of a moving hoist cable, with non-homogenous
boundary conditions. Secondly, to be able to reduce the PDE model
to an ODE model using a Galerkin reduction method, we introduce a
change of variables and re-write the first PDE model in a new
coordinates, where the new PDE model has zero boundary conditions.
Let us first enumerate the assumptions under which our model is
valid: 1) The elevator ropes are modelled within the framework of
string theory, 2) The elevator car is modelled as a point mass, 3)
The vibration in the second lateral direction is not included, 4)
The suspension of the car against its guide rails is assumed to be
rigid, 5) The mass of the semi-active damper is considered to be
negligible compared to the elevator car mass.

\begin{figure}
\vspace{-0.5cm}
  \begin{center}
\hspace{-0.8cm}\vspace{-0.5cm}
\epsfig{figure=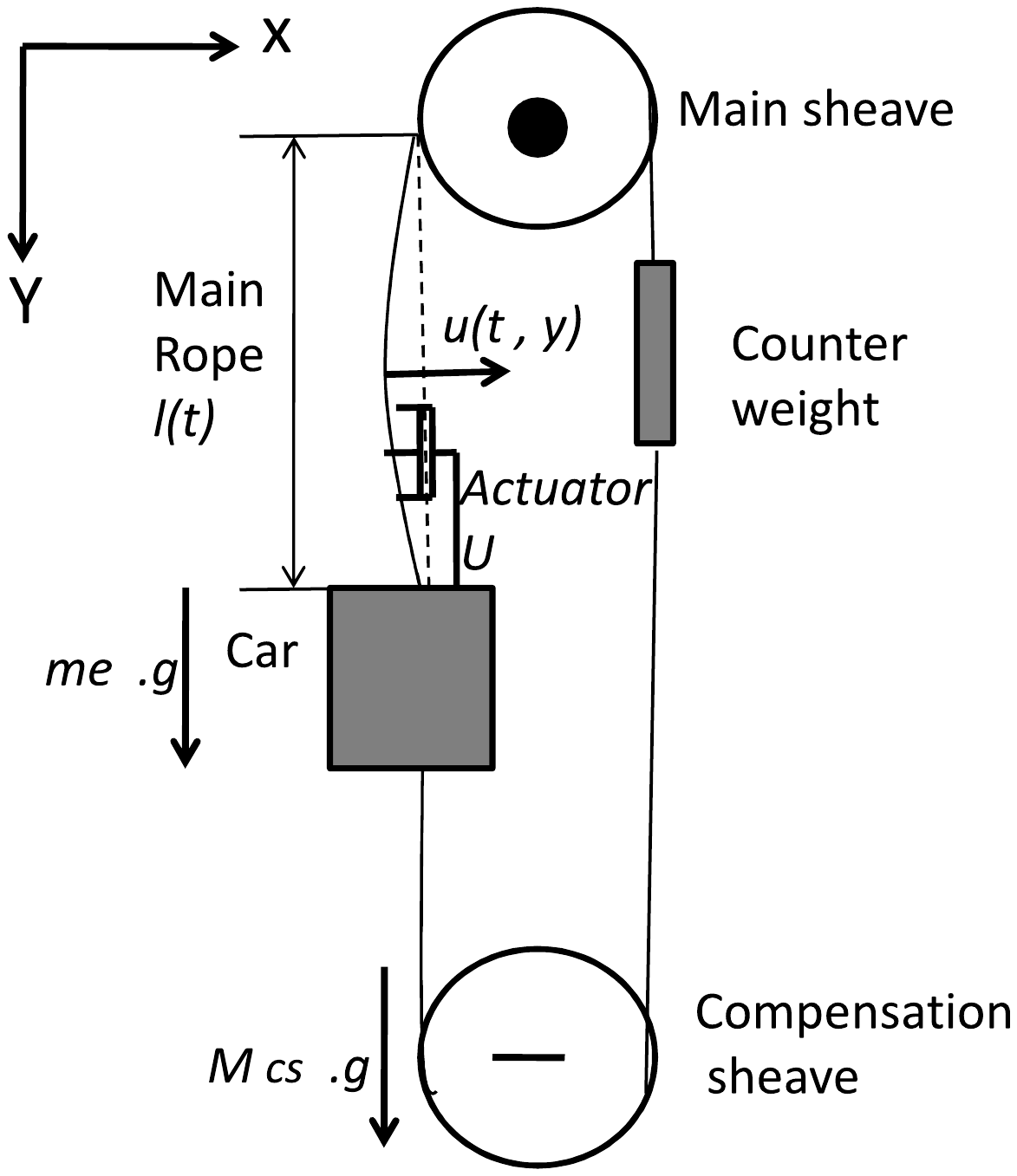,height=7cm,width=10cm}
\caption{Schematic representation of an elevator shaft showing the
different variables used in the model}

    \label{fig1}
 \end{center}
 \end{figure}
Under the previous assumption, following \cite{ZX03,K011} and
\cite{ZC06}, the general PDE model of an elevator rope, depicted
on Figure \ref{fig1}, is given by
\begin{equation}\label{PDE1}
\begin{array}{l}
\rho(\frac{\partial^2}{\partial
t^2}+v^{2}(t)\frac{\partial^2}{\partial
y^2}+2v(t)\frac{\partial}{\partial y\partial
t}+a\frac{\partial}{\partial
y}\big)u(y,t)\\-\frac{\partial}{\partial y}T(y,t)\frac{\partial
u(y,t)}{\partial y}+c_{p}\big(\frac{\partial}{\partial
t}+v(t)\frac{\partial}{\partial
y}\big)u(y,t)\\+k_{dp}\big(\frac{\partial}{\partial
t}+v(t)\frac{\partial}{\partial y}\big)u(y,t)\delta(y-l+l_{dp})=0
\end{array}
\end{equation}
where $u(y,t)$ is the lateral displacement of the rope. $\rho$ is
the mass of the rope per unit length. $T$ is the tension in the
rope, which varies depending on which rope in the elevator system
we are modelling, i.e. main rope, compensation rope, etc. $c_{p}$
is the damping coefficient of the rope per unit length.
$v=\frac{\partial l(t)}{\partial t}$ is the elevator rope
velocity, where $l:\;\mathbb{R}\rightarrow\mathbb{R}$ is a
function (at least $C^{2}$ ) modelling the time-varying rope
length. $a=\frac{\partial^2 l(t)}{\partial t^2}$ is the elevator
rope acceleration. The last term in the righthand-side of
(\ref{PDE1}) has been added to model the semi-active linear damper
effect on the rope at the contact point $\delta(y-l+l_{dp})$,
where $\delta$ is the Dirac impulse function, $k_{dp}$ is the
damping coefficient (the control variable) and $l_{dp}$ is the
distance between the car-top and the
point of attach of the damper to the controlled rope, e.g. \cite{ZC06}.\\
The PDE (\ref{PDE1}) is associated with the following two boundary
conditions:
\begin{equation}\label{BC1}
\begin{array}{l}
u(0,t)=f_{1}(t)\\
u(l(t),t)=f_{2}(t)
\end{array}
\end{equation}
where $f_1(t)$ is the time varying disturbance acting on the rope
at the level of the machine room, due to external disturbances,
e.g. wind gust. $f_2(t)$ is the time varying disturbance acting at
the level of the car, due to external disturbances. In this work
we assume that the two boundary disturbances acting on the rope
are related via the relation:
\begin{equation}\label{BC_relation}
f_{2}(t)=f_{1}(t)sin\big(\frac{\pi(H-l)}{2H}\big),\;H\in\mathbb{R}
\end{equation}
where $H$ is the height of the building. This expression is an
approximation of the propagation of the boundary disturbance
$f_{1}$ along the building structure, based on the length $l$, it
leads to $f_{2}=f_{1}$ for a length $0$ (which is expected), and a
decreasing force along the building until is vanishes at $l=H$,
$f_{2}=0$ (which makes sense, since the effect of any disturbance
$f_{1}$, for example wind gusts, is expected to vanish at the
bottom of the building). As we mentioned earlier the tension of
the rope $T(y)$ depends on the type of the rope that we are
dealing with. In the sequel, we concentrate on the main rope of
the elevator, the remaining ropes are modelled using the same
steps by simply changing the rope
tension expression.\\
For the case of the main rope, the tension is given by
\begin{equation}\label{tension1}
T(y,t)=(m_e+\rho(l(t)-y))(g-a(t))+0.5 M_{cs}g
\end{equation}
where $g$ is the standard gravity constant, $m_e, M_{cs}$ are the
mass of the car and the compensating sheave, respectively.
Next, we reduce the PDE model (\ref{PDE1}) to a more tractable
model for control, using a projection Galerkin method or assumed
mode approach, e.g. \cite{M67,HLB96}.\\To be able to apply the
assumed mode approach, let us first apply the following one-to-one
change of coordinates to the equation (\ref{PDE1})
\begin{equation}\label{change_of_coordinates}
u(y,t)=w(y,t)+\frac{l(t)-y}{l(t)}f_{1}(t)+\frac{y}{l(t)}f_{2}(t)
\end{equation}
One can easily see that this change of coordinates implies trivial
boundary conditions
\begin{equation}\label{BC2}
\begin{array}{l}
w(0,t)=0\\
w(l(t),t)=0
\end{array}
\end{equation}
After some algebraic and integral manipulations, the PDE model
(\ref{PDE1}) writes in the new coordinates as
\begin{equation}\label{PDE3}
\begin{array}{c}
\rho\frac{\partial ^2 w}{\partial t^2}+2v(t)\rho\frac{\partial^2
w}{\partial y\partial t}+\left(\rho
v^2-T(y,t)\right)\frac{\partial^2 w}{\partial
y^2}\\+(G(t)+vk_{dp}\delta(x-l+l_{dp}))\frac{\partial w}{\partial
y}\\+(c_{p}+k_{dp}\delta(x-l+l_{dp}))\frac{\partial w}{\partial
t}\\=y\left(-\rho s_{1}(t)-c_{p} s_{2}(t)\right)-\rho
f^{(2)}_{1}+s_{4}(t)\\-\frac{\partial w}{\partial
t}\big(\frac{l-y}{l}f_{1}(t)+\frac{y}{l}f_{2}(t)\big)k_{dp}\delta(x-l+l_{dp})\\-vk_{dp}\frac{f_{2}-f_{1}}{l}\delta(x-l+l_{dp})
\end{array}
\end{equation}
where $ G(t)=\rho a(t)-\frac{\partial T}{\partial y}+c_{p} v(t)$,
and the $s_{i}$ variables are defined as
\begin{equation}\label{S_i}
\begin{array}{l}
s_{1}(t)=\frac{ll^{(2)}-2\dot{l}^{2}}{l^{3}}f_{1}(t)+2\frac{\dot{l}}{l^{2}}\dot{f}_{1}\\+\frac{(l^{3}f^{(2)}_{2}-f_{2}l^{2}l^{(2)}+2l\dot{l}^{2}f_{2}-2l^{2}\dot{l}\dot{f}_{2})}{l^4}-\frac{f^{(2)}_{1}}{l}\\
s_{2}(t)=\frac{\dot{l}}{l^{2}}f_{1}-\frac{\dot{f}_{1}}{l}+\frac{\dot{f}_{2}}{l}-f_{2}\frac{\dot{l}}{l^{2}}\\
s_{3}(t)=\frac{f_{2}-f_{1}}{l}\\
s_{4}(t)=-2 v(t)\rho s_{2}(t)-G(t)s_{3}(t)-c_{p}\dot{f}_{1}(t)
\end{array}
\end{equation}

 associated with the two-point boundary conditions
\begin{equation}\label{BC3}
w(0,t)=0,\;\;w(l(t),t)=0.
\end{equation}
Now instead of dealing with the PDE (\ref{PDE1}) with non-zero
boundary conditions, we can use the equivalent model, given by
equation (\ref{PDE3}) associated
with trivial boundary conditions (\ref{BC3}).\\
 Following the assumed-modes technique, the solution of the equation
 (\ref{PDE3}), (\ref{BC3}) writes as
 \begin{equation}\label{discretization_1}
w(y,t)=\sum_{j=1}^{j=N}q_{j}(t)\phi_{j}(y,t),\;\;N\in\mathbb{N}
 \end{equation}
 where $N$ is the number of bases (modes), included in the
 discretization, $\phi_{j},\;j=1,...,N$ are the discretization
 bases and $q_{j},\;\;j=1,..,N$ are the discretization
 coordinates. In order to simplify the analytic manipulation of
 the equations, the base functions are chosen to satisfy the
 following normalization constraints
 \small\begin{equation}\label{normalization_1}
\int_{0}^{l(t)}\phi_{j}^{2}(y,t)dy=1,\;
\int_{0}^{l(t)}\phi_{i}(y,t)\phi_{j}(y,t)dy=0,\;\;\forall i\neq j
 \end{equation}
 To further simplify the base functions, we define the normalized
 variable, e.g. \cite{ZC06,ZX03}
$\xi(t)=\frac{y(t)}{l(t)}$,
 and the normalized base functions
$\phi_{j}(y,t)=\frac{\psi_{j}(\xi)}{\sqrt{l(t)}},\;\;j=1,...,N$.
 In these new coordinates the normalization constraints write as
$\int_{0}^{1}\psi_{j}^{2}(\xi)d\xi=1,\;
\int_{0}^{1}\psi_{i}(\xi)\psi_{j}(\xi)d\xi=0,\;\;\forall i\neq j$.
After discretization of the PDE-based model (\ref{PDE3}),
(\ref{S_i}) and (\ref{BC3}) (e.g. refer to \cite{ZX03}), we can
write the reduced ODE-model based on $N$-modes as
\begin{equation}\label{ODE_model}
M\ddot{q}+(C+\tilde{C}U)\dot{q}+(K+\tilde{K}
U)q=F(t)+\tilde{F}(t)U,\;q\in\mathbb{R}^{N}
\end{equation}
where
\begin{equation}\label{ODE_model_coefficients}
\begin{array}{l}
{\scriptstyle U=k_{dp}}\\
{\scriptstyle{M_{ij}=\rho\delta_{ij}}}\\
{\scriptstyle{C_{ij}=\rho l^{-1}\dot{l}\left(2\int_{0}^{1}(1-\xi)\psi_{i}(\xi)\psi^{'}_{j}(\xi)d\xi-\delta_{ij}\right)+c_{p}\delta_{ij}}}\\
{\scriptstyle \tilde{C}_{ij}=l^{-1}\psi_{i}(\frac{l-l_{dp}}{l}) \psi_{j}(\frac{l-l_{dp}}{l})}\\
{\scriptstyle{K_{ij}=\frac{1}{4}\rho
l^{-2}\dot{l}^{2}\delta_{ij}-\rho
l^{-2}\dot{l}^{2}\int_{0}^{1}(1-\xi)^{2}\psi^{'}_{i}(\xi)\psi^{'}_{j}(\xi)d\xi}}\\{\scriptstyle{+\rho
l^{-1}(g-a(t))\int_{0}^{1}(1-\xi)\psi^{'}_{i}(\xi)\psi^{'}_{j}(\xi)d\xi+}}
{\scriptstyle{m_{e}l^{-2}(g-a(t))\int_{0}^{1}\psi^{'}_{i}(\xi)\psi^{'}_{j}(\xi)d\xi}}\\{\scriptstyle{+\rho(l^{-2}\dot{l}^{2}-l^{-1}\ddot{l})\left(0.5\delta_{ij}
-\int_{0}^{1}(1-\xi)\psi_{i}(\xi)\psi^{'}_{j}(\xi)d\xi\right)-}}\\
{\scriptstyle{c_{p}\dot{l}l^{-1}(\int_{0}^{1}\psi_{i}(\xi)\psi^{'}_{j}(\xi)\xi
d\xi+0.5\delta_{ij})+0.5M_{cs}gl^{-2}\int_{0}^{1}\psi^{'}_{i}(\xi)\psi^{'}_{j}(\xi)d\xi}}\\
{\scriptstyle \tilde{k}_{ij}=l^{-2}\dot{l}\big(-\psi^{'}_{j}(\frac{l-l_{dp}}{l})\psi_{i}(\frac{l-l_{dp}}{l})(\frac{l-l_{dp}}{l})-0.5\psi_{i}(\frac{l-l_{dp}}{l})\psi_{j}(\frac{l-l_{dp}}{l})}\\
{\scriptstyle+\psi^{'}_{j}(\frac{l-l_{dp}}{l})\psi_{i}(\frac{l-l_{dp}}{l})\big)}\\
{\scriptstyle{ F_{i}(t)=-l\sqrt{l}\left(\rho
s_{1}(t)+c_{p}s_{2}(t)\right)\int_{0}^{1}\psi_{i}(\xi)\xi
d\xi}}\\{\scriptstyle{+\sqrt{l}\left(s_{4}(t)-\rho
f^{(2)}_{1}(t)\right)\int_{0}^{1}\psi_{i}(\xi)d\xi}}\\
{\scriptstyle
\tilde{F}_{i}(t)=\frac{\dot{f}_{1}}{\sqrt{l}}\psi_{i}(\frac{l-l_{dp}}{l})+\frac{l-l_{dp}}{\sqrt{l}}\psi_{i}(\frac{l-l_{dp}}{\sqrt{l}})(\dot{l}l^{-2}(f_{1}-f_{2})+l^{-1}(\dot{f}_{2}-\dot{f}_{1}))}
\\
\delta_{ij}=\left\{\begin{array}{l}0,\;i\neq
j\\1,\;i=j\end{array}\right.
\end{array}
\end{equation}
where $i,j\in\{1,...,N\}$, and $s_{k},\;k=1,2,3,4$ are given in
(\ref{S_i}).
\begin{remark}
The model (\ref{ODE_model}), (\ref{ODE_model_coefficients}) has
been obtained for the general case of time-varying rope length
$l(t)$, however, in this paper we only consider the case of
stationary ropes $l=cte$, which is directly deduced from
(\ref{ODE_model}), (\ref{ODE_model_coefficients}), by setting
$\dot{l}=\ddot{l}=0,\;\forall t$. The case where the car is static
and the control system has to reject the ropes' oscillations is of
interest in practical setting. Indeed, besides the case of
commercial buildings at night (where the elevators are not in
use), in many situations where the building is swaying due to
external strong weather conditions, the elevators are stopped, for
the security of the passengers. The control system is then used to
damp out the ropes sway, to avoid the ropes from damaging the
elevator system, and be still functional after the external
disturbances have passed. \footnote{ Of course, there are also
practical cases were the cars are in motion and an external
disturbance occurs. These cases correspond to a time-varying rope
length, which we have also studied, however, due to space
limitation, we could not include all the results in this paper.
The case of time-varying rope length will be presented in another
report.}
\end{remark}

\section{Main result: Lyapunov-based semi-active damper control}\label{section2}
The first controller deals with the case where the building,
hosting a stationary elevator (stopped at a given floor), sustains
a brief (impulse-like) external disturbance. For example, an
earthquake impulse with a sufficient force to make the top of the
building oscillate, or a strong wind gust that happens over a
short period of time, exciting the building structure and implying
residual vibrations of the building even after the wind gust
interruption. In these cases, the elevator ropes will vibrate,
starting from a non-zero initial conditions, due to the
impulse-like external disturbances (i.e., happening over a short
time interval), and this case correspond to the model
(\ref{ODE_model}), (\ref{ODE_model_coefficients}) with non-zero
initial conditions and zero external disturbances. We can now
state the following
theorem.\\\\
 \begin{thm}
Consider the rope dynamics (\ref{ODE_model}),
(\ref{ODE_model_coefficients}), with non-zero initial conditions,
with no external disturbances, i.e., $f_{1}(t)=f_{2}(t)=0,\forall
t$, then the feedback control
\begin{equation}\label{control_1_t}
U(z)=u_{max}\frac{\dot{q}^{T}\tilde{C}\dot{q}
}{\sqrt{1+(\dot{q}^{T}\tilde{C}\dot{q})^{2}}}
\end{equation}
where $z=(q^{T},\dot{q}^{T})^{T}$, implies that $q(t)\rightarrow
0,\;for\; t\rightarrow 0$, furthermore $|U|\leq u_{max},\;\forall
t$, and $|U|$ deceases with the decrease of $\dot{q}^{T}\tilde{C} \dot{q}$.\\\\
\end{thm}
\begin{proof}
 We define the control Lyapunov function  as
\begin{equation}\label{lyapunov1-t}
V(z)=\frac{1}{2}\dot{q}^{T}(t)M\dot{q}(t)+\frac{1}{2}q^{T}(t)K(t)q(t)
\end{equation}
where $z=(q^{T},\dot{q}^{T})^{T}$.\\
 First we compute the derivative of the Lyapunov function
along the dynamics  (\ref{ODE_model}), without disturbances, i.e.,
$F(t)=0,\;\tilde{F}(t)=0 \forall t$
\begin{equation}\label{lyapunov_dot_2}
\begin{array}{l}
\dot{V}(z)=\dot{q}^{T}(-C\dot{q}-\tilde{C}U\dot{q}-Kq)+q^{T}K\dot{q}\\
\hspace{0cm}=-\dot{q}^{T}C\dot{q}-\dot{q}^{T}\tilde{C}\dot{q} U
\end{array}
\end{equation}
Next, using $U$ defined in (\ref{control_1_t}), we have
\begin{equation}\label{lyapunov_dot_4}
\begin{array}{l}
\dot{V}(z)\leq -u_{max}\frac{(\dot{q}^{T}\tilde{C}\dot{q})^{2}
}{\sqrt{1+(\dot{q}^{T}\tilde{C}\dot{q})^{2}}}\\
\end{array}
\end{equation}
Using LaSalle theorem, e.g. \cite{K96}, and the fact that
$\tilde{C}$ is symmetric positive definite we can conclude that
the states of the closed-loop dynamics converge to the set
$S=\{z=(q^{T},\dot{q}^{T})^{T}\in\mathbb{R}^{2N},\;s.t.\;\dot{q}=0\;\}$.
Next, we analyze the closed-loop dynamics: Since the boundedness
of $V$ implies boundedness of $\dot{q},\;q$ and by equation
(\ref{ODE_model}), boundedness of $\ddot{q}$. Boundedness of
$\dot{q},\;\ddot{q}$ implies the uniform continuity of
$q,\;\dot{q}$, which again by (\ref{ODE_model}), implies the
uniform continuity of $\ddot{q}$. Next, since $\dot{q}\rightarrow
0$, and using Barbalat's Lemma, e.g. \cite{K96}, we conclude that
$\ddot{q}\rightarrow 0$, and by invertibility of the stiffness
matrix $K+\beta U$ we conclude that $q\rightarrow 0$. Finally, the
fact that $V$ is a radially unbounded function, ensures that the
equilibrium point $(q,\dot q)=(0,0)$ is globally asymptotically
stable. Furthermore the fact that $|U|\leq u_{max}$, and the
decrease of $|U|$ as function of $\dot{q}^{T}\tilde{C} \dot{q}$ is
deduced from equation (\ref{control_1_t}), since
$\frac{\dot{q}^{T}\tilde{C}\dot{q}
}{\sqrt{1+(\dot{q}^{T}\tilde{C}\dot{q})^{2}}}\leq 1$ and
$\frac{1}{\sqrt{1+(\dot{q}^{T}\tilde{C}\dot{q})^{2}}}\leq 1$.
\end{proof}
\begin{remark}
It is clear from equation (\ref{lyapunov_dot_2}) that the trivial
choice of a constant positive damping control $U$, will also imply
a convergence of the sway dynamics to zero. However, the
controller (\ref{control_1_t}) has the advantage to  require less
control energy comparatively to a constant damping, since by
construction of the control law (\ref{control_1_t}), the
stabilizing damping force decreases together with the decrease of
the sway.
\end{remark}

The controller $U$ given by (\ref{control_1_t}) does not take into
account the disturbance $F(t)$ explicitly. Next, we present a
controller which deals with the case of a static elevator in a
building under sustained external disturbances, e.g. sustained
wind forces on commercial buildings at night where the cars are
static. In this case $F(t)\neq 0,\;\tilde{F}(t)\neq 0$ over a
non-zero time interval, and satisfy the following
assumption.\\
\begin{assumption}\label{assumption0}
The time varying disturbance functions $f_{1},\;f_{2}$ are such
that, the functions $F(t),\;\tilde{F}(t)$ are bounded, i.e.
$\exists (F_{max},\;\tilde{F}_{max}),\;s.t.\;|F(t)|\leq
F_{max},\;|\tilde{F}(t)|\leq \tilde{F}_{max},\;\forall t$.
\end{assumption}
\begin{thm}\label{theorem3}
Consider the rope dynamics (\ref{ODE_model}),
(\ref{ODE_model_coefficients}), under non-zero external
disturbances, i.e., $f_{1}(t)\neq 0,\;f_{2}(t)\neq 0$ satisfying
Assumption \ref{assumption0}, with the feedback control
\begin{equation}\label{control_2_t}
\begin{array}{l}
{\scriptstyle U(z)=u_{max_p}\frac{\dot{q}^{T}\tilde{C}\dot{q}
}{\sqrt{1+(\dot{q}^{T}\tilde{C}\dot{q})^{2}}}}\\{\scriptstyle
+\dot{q}^{T}\tilde{C}\dot{q}}{\scriptstyle\frac{v_{1max}(\tilde{F}_{max}|\dot{q}|v_{1max}+\tilde{F}_{max}|\dot{q}|u_{max_p})}{\sqrt{1+(\dot{q}^{T}\tilde{C}\dot{q})^{2}(\tilde{F}_{max}|\dot{q}|v_{1max}+\tilde{F}_{max}|\dot{q}|u_{max_p})^{2}}}}\\
{\scriptstyle+\dot{q}^{T}\tilde{C}\dot{q}\frac{v_{2max}(|\dot{q}|F_{max}+v_{2max}\tilde{F}_{max}|\dot{q}|)}{\sqrt{1+(\dot{q}^{T}\tilde{C}\dot{q})^{2}(|\dot{q}|F_{max}+v_{2max}\tilde{F}_{max}|\dot{q}|)^{2}}}}
\end{array}
\end{equation}
where $u_{max_p},\;v_{1max},\;v_{2max}>0$, are chosen s.t.
$(u_{max_p}+v_{1max}+v_{2max})\leq u_{max}$ and
$z=(q^{T},\dot{q}^{T})^{T}$. Then, if we define the two invariant
sets: $$\begin{array}{l}
S_{1}=\{(q^{T},\dot{q}^{T})^{T}\in\mathbb{R}^{2N},\;s.t.\;\\\frac{(\dot{q}^{T}\tilde{C}\dot{q})^{2}}{\sqrt{1+(\dot{q}^{T}\tilde{C}\dot{q})^{2}(\tilde{F}_{max}|\dot{q}|v_{1max}+\tilde{F}_{max}|\dot{q}|u_{max_p})^{2}}}\\\leq\frac{1}{v_{1max}}\}\end{array}$$,
and
$$\begin{array}{l}S_{2}=\{(q^{T},\dot{q}^{T})^{T}\in\mathbb{R}^{2N},\;s.t.\;\\\frac{(\dot{q}^{T}\tilde{C}\dot{q})^{2}}{\sqrt{1+(\dot{q}^{T}\tilde{C}\dot{q})^{2}(|\dot{q}|F_{max}+v_{2max}\tilde{F}_{max}|\dot{q}|)^{2}}}\leq\frac{1}{v_{2max}}\}\end{array}$$,
the controller (\ref{control_2_t}) ensures that the state vector
$z$ converges to the the invariant set $S_{1}$ or $S_{2}$, and
that $|U|\leq u_{max},\;\forall t$.
\end{thm}
\begin{proof}
Let us consider again the Lyapunov function (\ref{lyapunov1-t}).
Its derivative along the dynamics (\ref{ODE_model}), with non-zero
disturbance, i.e. $F(t)\neq 0,\;\tilde{F}\neq 0$, writes as
\begin{equation}\label{lyapunov_dot_3}
\begin{array}{l}
\dot{V}(x)=\dot{q}^{T}(-C\dot{q}-\tilde{C}U\dot{q}-Kq+\tilde{F}U)\\+q^{T}K\dot{q}+\dot{q}^{T}F(t)\\
\hspace{0cm}=-\dot{q}^{T}C\dot{q}-\dot{q}^{T}\tilde{C}U\dot{q}+\dot{q}^{T}\tilde{F}U
+\dot{q}^{T}F(t)
\\
\leq -\dot{q}^{T}\tilde{C}U\dot{q}+\dot{q}^{T}\tilde{F}U
+\dot{q}^{T}F(t)
\end{array}
\end{equation}
Now, we use the concept of Lyapunov reconstruction, e.g.
\cite{BL09-4}, we write the control
\begin{equation}\label{total_control}
U=u_{nom_p}+v_{1}+v_{2}
\end{equation}
where, $u_{nom_p}$ is the nominal controller, given by
(\ref{control_1_t}) with $u_{max}=u_{max_p}$, designed for the
case where $F(t)=\tilde{F}(t)=0,\;\forall t$. The remaining terms
$v_{1},\;v_{2}$ are added to compensate for the effect of the
disturbances $\tilde{F}$ and $F$, respectively. We design $v_{1}$
and $v_{2}$ in two steps:\\
- First step: We assume that $\tilde{F}\neq 0,\;F(t)=0,\;\forall
t$, and design $v_{1}$ to compensate for the effect of
$\tilde{F}$. In this case $U=u_{nom_p}+v_{1}$.\\In this case, the
Lyapunov function derivative is bounded as as follows
$$
\begin{array}{l} \dot{V}(x)\leq
-\dot{q}^{T}\tilde{C}(u_{nom_p}+v_{1})\dot{q}+\dot{q}^{T}\tilde{F}(u_{nom_p}+v_{1})
\\
\leq
-\dot{q}^{T}\tilde{C}v_{1}\dot{q}+\dot{q}^{T}\tilde{F}(u_{nom_p}+v_{1})
\\\leq -\dot{q}^{T}\tilde{C}\dot{q}v_{1}+\dot{q}^{T}\tilde{F}u_{nom_p}+\dot{q}^{T}\tilde{F}v_{1}
\end{array}
$$
 which under Assumption \ref{assumption0}, gives
$$
\begin{array}{l}
\leq
-\dot{q}^{T}\tilde{C}\dot{q}v_{1}+|\dot{q}|\tilde{F}_{max}u_{max_p}+|\dot{q}|\tilde{F}_{max}v_{1max}
\end{array}
$$
Now, if we define (to simplify the notations), the term
$$T_{1}=+|\dot{q}|\tilde{F}_{max}u_{max_p}\\+|\dot{q}|\tilde{F}_{max}v_{1max}$$
and if we choose the controller
$$v_{1}=\frac{v_{1max}T_{1}\dot{q}^{T}\tilde{C}\dot{q}}{\sqrt{1+T_{1}^{2}(\dot{q}^{T}\tilde{C}\dot{q})^{2}}}$$
This leads to the following Lyapunov function derivative upper
bound
$$
\begin{array}{l}
\leq
T_{1}\big(1-\frac{v_{1max}(\dot{q}^{T}\tilde{C}\dot{q})^{2}}{\sqrt{1+T_{1}^{2}(\dot{q}^{T}\tilde{C}\dot{q})^{2}}}\big)=B_{1}
\end{array}
$$
- Second step: We assume that $\tilde{F}(t)\neq 0, \;F(t)\neq 0$,
and design $v_{2}$ to compensate for $F$. In this case we write
the total control as $U=u_{nom_p}+v_{1}+v_{2}$. Now the
upper-bound of the Lyapunov function derivative (along the total
control) writes as
$$
\begin{array}{l}
\leq
B_{1}-\dot{q}^{T}\tilde{C}\dot{q}v_{2}+\dot{q}^{T}\tilde{F}v_{2}+\dot{q}^{T}F\\
\leq
B_{1}-\dot{q}^{T}\tilde{C}\dot{q}v_{2}+|\dot{q}|\tilde{F}_{max}v_{2max}+|\dot{q}|F_{max}
\end{array}
$$
Next, if we define the term
$$T_{2}=|\dot{q}|\tilde{F}_{max}v_{2max}+|\dot{q}|F_{max}$$
and if we choose the controller
$$v_{2}=\frac{v_{2max}T_{2}\dot{q}^{T}\tilde{C}\dot{q}}{\sqrt{1+T_{2}^{2}(\dot{q}^{T}\tilde{C}\dot{q})^{2}}}$$
This leads to the following Lyapunov function derivative
upper-bound
$$
\begin{array}{l}
\leq
B_{1}+T_{2}\big(1-\frac{v_{2max}(\dot{q}^{T}\tilde{C}\dot{q})^{2}}{\sqrt{1+T_{2}^{2}(\dot{q}^{T}\tilde{C}\dot{q})^{2}}}\big)\\
\leq B_{1}+B_{2}
\end{array}
$$
where
$B_{2}=T_{2}\big(1-\frac{v_{2max}(\dot{q}^{T}\tilde{C}\dot{q})^{2}}{\sqrt{1+T_{2}^{2}(\dot{q}^{T}\tilde{C}\dot{q})^{2}}}\big)$.
\\
By choosing $v_{1max},\;v_{2max}$ high enough the two terms
$B_{1},\;B_{2}$ will be made negative. We can then analyze two
cases:\\
1- First, the trajectories keep decreasing until they reach the
invariant set
$$S_{1}=\{(q^{T},\dot{q}^{T})^{T}\in\mathbb{R}^{2N},\;s.t.\;B_{1}\geq
0\}$$ Then, the trajectories can either keep decreasing (if
$|B_{2}|>B_{1}$) until they enter the invariant set
$$S_{2}=\{(q^{T},\dot{q}^{T})^{T}\in\mathbb{R}^{2N},\;s.t.\;B_{2}\geq
0\}$$or they get stuck at $S_{1}$.\\
2- Second, the trajectories decrease until they reach the
invariant set $S_{2}$ first, and stay there, or keep decreasing
until they reach the invariant set $S_{1}$.\\In both cases, the
trajectories will end up in either $S_{1}$ or $S_{2}$. Finally,
the fact that $|U|\leq u_{max}$ is obtained by construction of the
three terms, since from (\ref{control_1_t}), we can write
$|u_{nom_p}|\leq u_{max_p}$ and by construction $|v_{1}|\leq
v_{1max}$, $|v_{2}|\leq v_{2max}$, which leads to $|U|\leq
u_{max_p}+v_{1max}+v_{2max}\leq u_{max}$.
\end{proof}
\begin{remark}
We want to underline here the fact that, contrary to the previous
case of impulse disturbances, in this case of sustainable external
disturbances, the use of a simple passive damper, i.e. $U=cte$,
might actually destabilize the system. Indeed, by examining
equation (\ref{lyapunov_dot_3}), we can see that if the $U$ is
constant the positive term $+\dot{q}^{T}\tilde{F}U$, which is due
the external disturbance, could overtake the damping negative term
$-\dot{q}^{T}\tilde{C}U\dot{q}$, leading to instability of the
system.
\end{remark}
\begin{remark}
The controllers (\ref{control_1_t}), (\ref{control_2_t}) are state
feedbacks based on $q,\;\dot{q}$, these states can be easily
computed from the sway measurements at $N$ given positions
$y(1),...,y(N)$, via equation (\ref{discretization_1}). The sway
$w(y,t)$ can be measured by laser displacement sensors placed at
the positions $y(i),\;i=1,2,...N$, along the rope,
e.g.\cite{OYNFN02}, subsequently $q$ can be computed by simple
algebraic inversion of (\ref{discretization_1}), and $\dot{q}$ can
be obtained by direct numerical differentiation of $q$.
\end{remark}
\section{Numerical example}\label{section3}
In this section we present some numerical results obtained on the
example presented in \cite{K011}.
\begin{table}\hspace*{-0.5cm}
\begin{center}
\begin{tabular}{|l|l|l|}
\hline {\bf Parameters} & {\bf Definitions} & {\bf Values}\\\hline
$n$ & Number of ropes &$ 8[-]$\\
$m_{e}$ & Mass of the car & $3500 [kg]$\\
$\rho$ & Main rope linear mass density & $2.11[kg/m]$\\
$l$ & Rope maximum length & $390[m]$\\
$H$ & Building height & $402.8[m]$\\
$c_{p}$ & Damping coefficient & $0.0315 [N.sec/m]$\\
\hline
\end{tabular}\caption{Numerical values of the mechanical parameters}
\label{table1}
\end{center}
\end{table}
The case of an elevator system with the mechanical parameters
summarized on Table \ref{table1} has been considered for the tests
presented hereafter. We write the controllers based on the model
(\ref{ODE_model}), (\ref{ODE_model_coefficients}) with one mode,
but we test them on a model with two modes. The fact is that one
mode is enough since when comparing the solution of the PDE
(\ref{PDE3}) to the discrete model (\ref{ODE_model}) the higher
modes shown to be negligible, and a discrete model with one mode
showed a very good match with the PDE model, but to make the
simulation tests more realistic we chose to test the controllers
on a two modes model. Furthermore, to make the simulation tests
more challenging we added a random white noise to the states fed
back to the controller (equivalent to about $\pm1\;cm$ of error on
the rope sway measurement from which the states are computed, see
Remark 1), and we filtered the control signal with a first order
filter with a cut frequency of $10\;hz$ and a delay term of $5$
sampling times, to simulate actuator dynamics and delays due to
signal transmission and computation time. First, to validate
Theorem 1, we present the results obtained by applying the
controller (\ref{control_1_t}), to the model (\ref{ODE_model}),
(\ref{ODE_model_coefficients}), with non-zero initial conditions
$q(0)=20,\;\dot{q}(0)=5$, and zero external disturbances, i.e.
$f_{1}(t)=f_{2}(t)=0,\;\forall t$. In these first tests, to show
the effect of the controller (\ref{control_1_t}) alone, without
the `help' of the system's natural damping, we fix the damping
coefficient to zero, i.e. $c_{p}=0$. We apply the controller
(\ref{control_1_t}), with $u_{max}=10^9\;N sec/m$. Figures
\ref{figure1}, \ref{figure1_zoom}\footnote{The figures' zoom is
included for the reader to have a better idea about the signals
shape.} show the rope sway obtained at half rope-length $y=195\;m$
with and without control. Without control the rope sway reaches a
maximum value of about $1.5\;m$. With control we see clearly the
expected damping effect of the controller, which reduces the sway
amplitude by half. The corresponding control force is depicted on
Figures \ref{figure2}, \ref{figure2_zoom}. We see that, as
expected from the theoretical analysis of Theorem 1, the control
force remains bounded by a maximum value of $40 k N sec/m$, which
is easily realizable by existing semi-active damper, e.g.
magnotorheological damper. Furthermore, as proven in Theorem 1,
the control amplitude decreases with the decrease of the sway.

\begin{figure}
\centering
\includegraphics[width=1\linewidth]{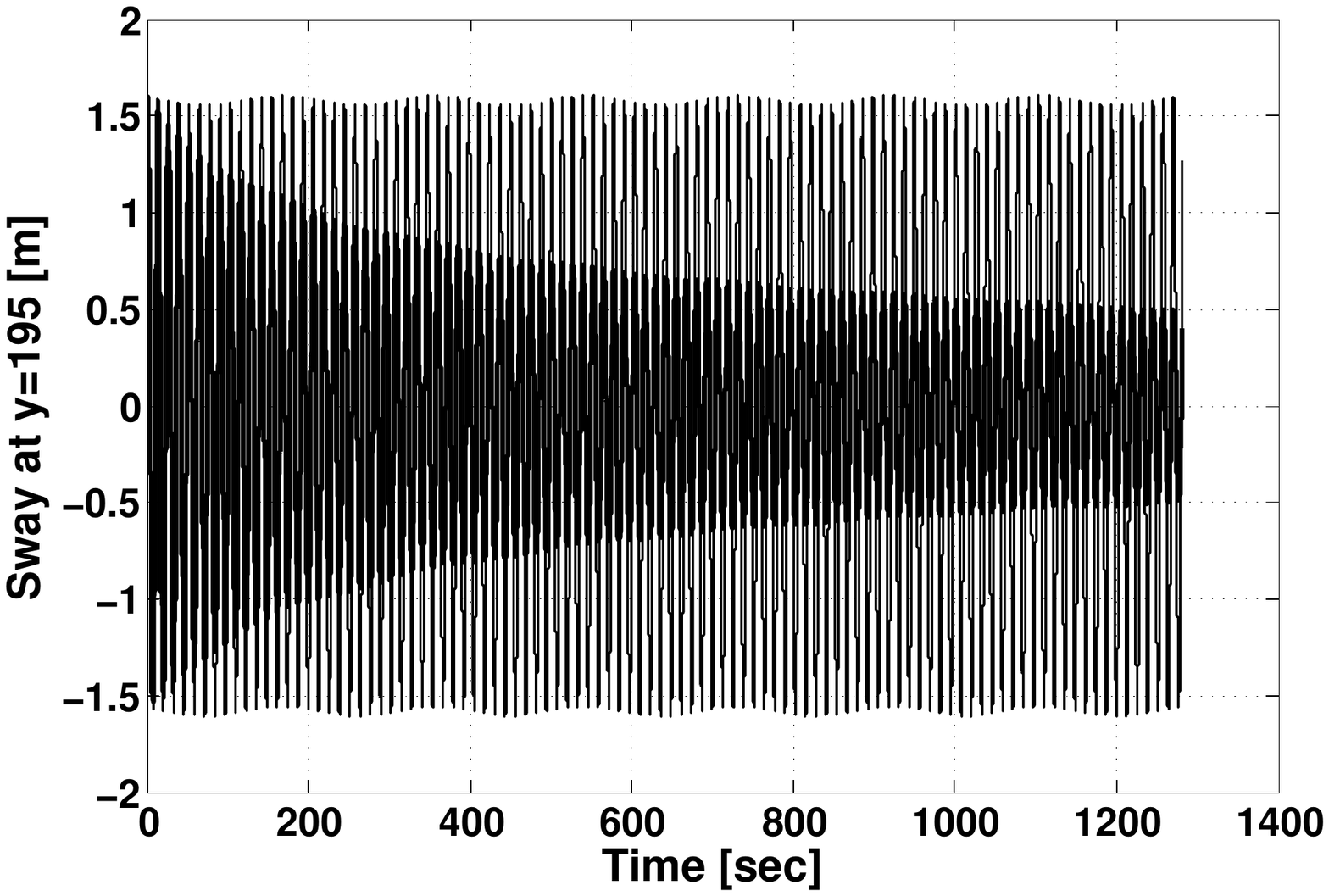}\vspace{-3cm}
\caption{Rope sway at $y=195\;m$: No control (thin line)- With
controller (\ref{control_1_t}) (bold line)}
\label{figure1}\vspace{-3cm}
\includegraphics[width=1\linewidth]{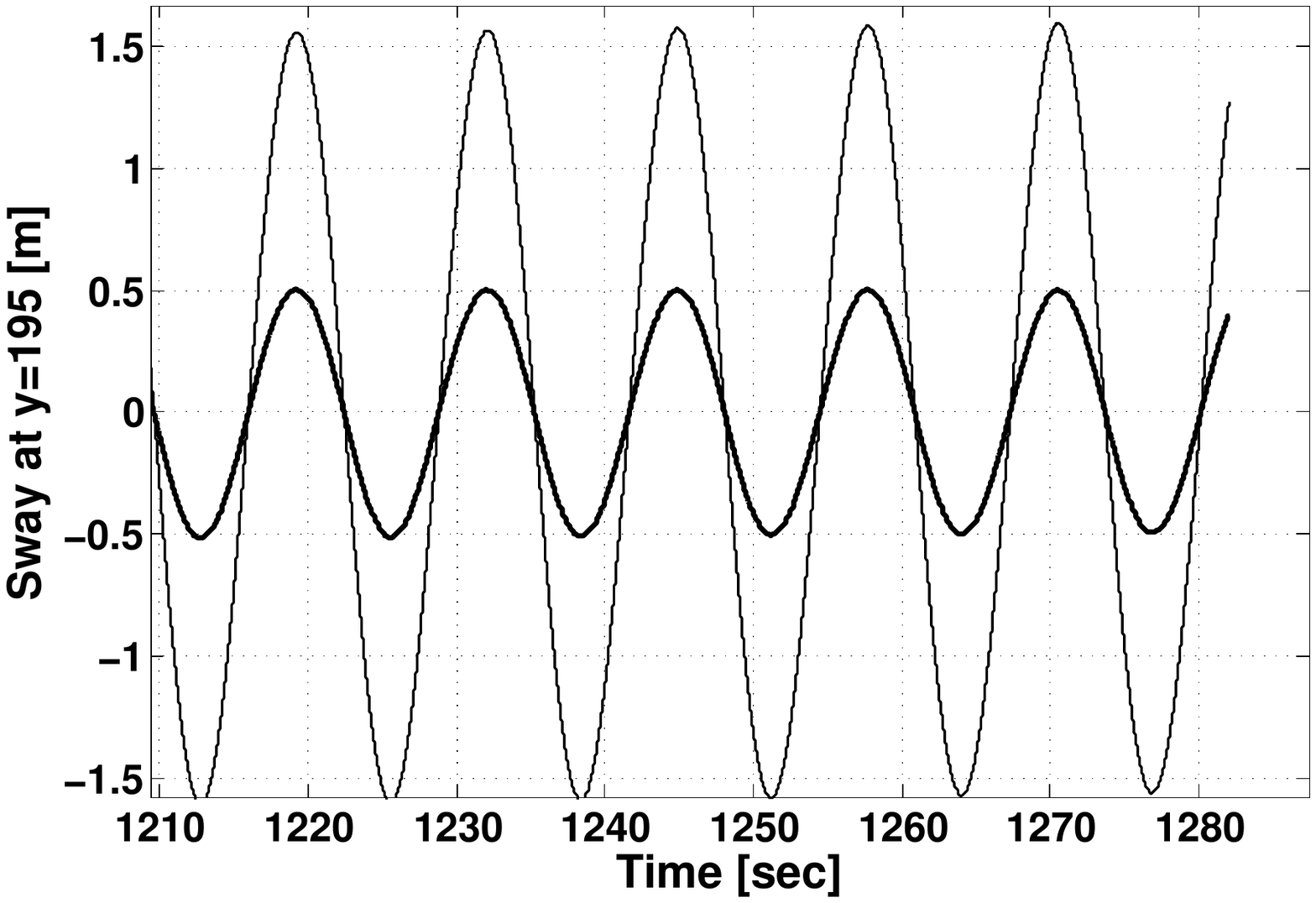}\vspace{-3cm}
\caption{Zoom of rope sway at $y=195\;m$: No control (thin line)-
With controller (\ref{control_1_t}) (bold line)}
\label{figure1_zoom}
\end{figure}
\begin{figure}
\centering
\includegraphics[width=1\linewidth]{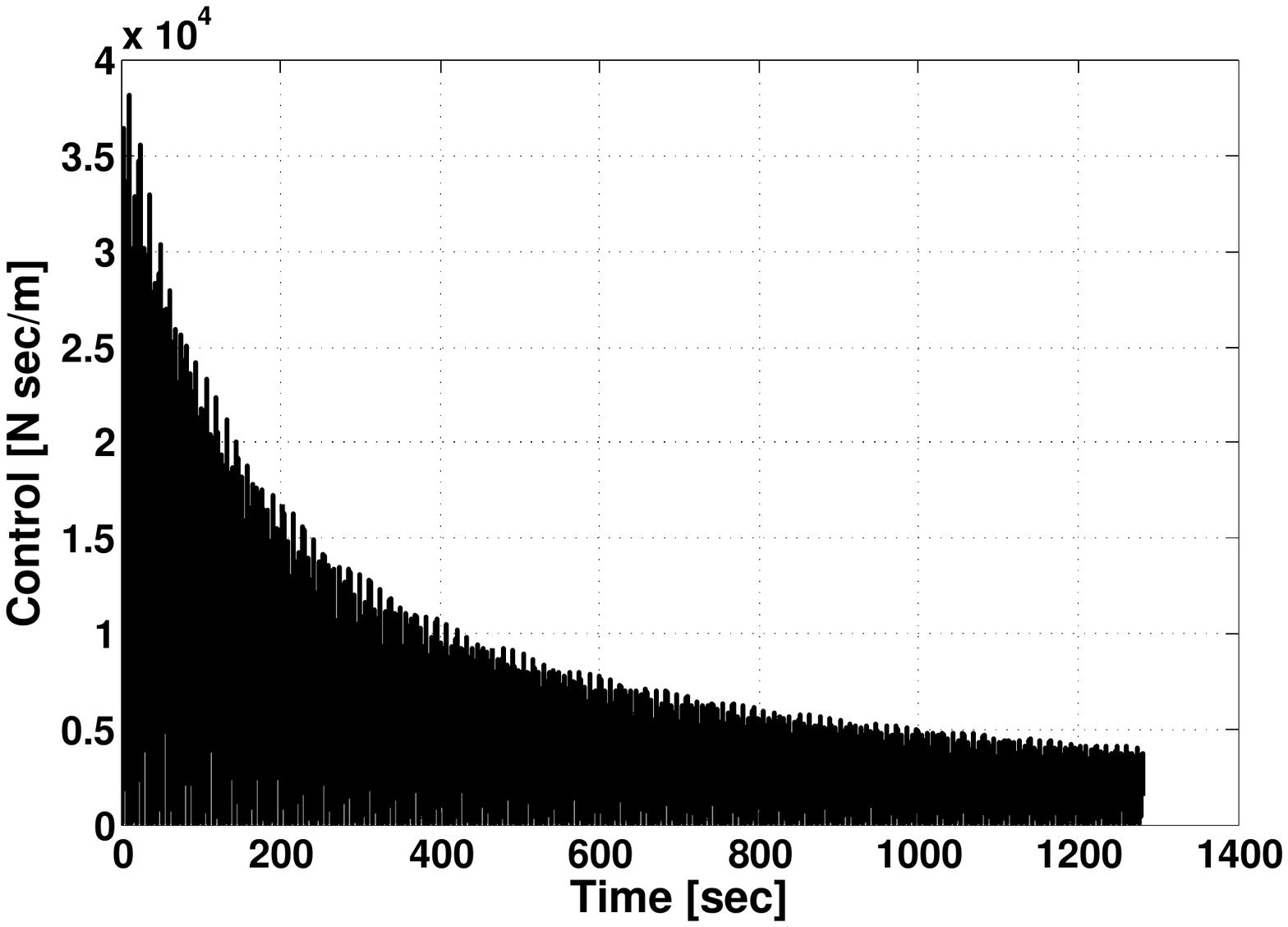}\vspace{-3cm}
\caption{Output of controller (\ref{control_1_t})} \label{figure2}
\vspace{-3cm}
\includegraphics[width=1\linewidth]{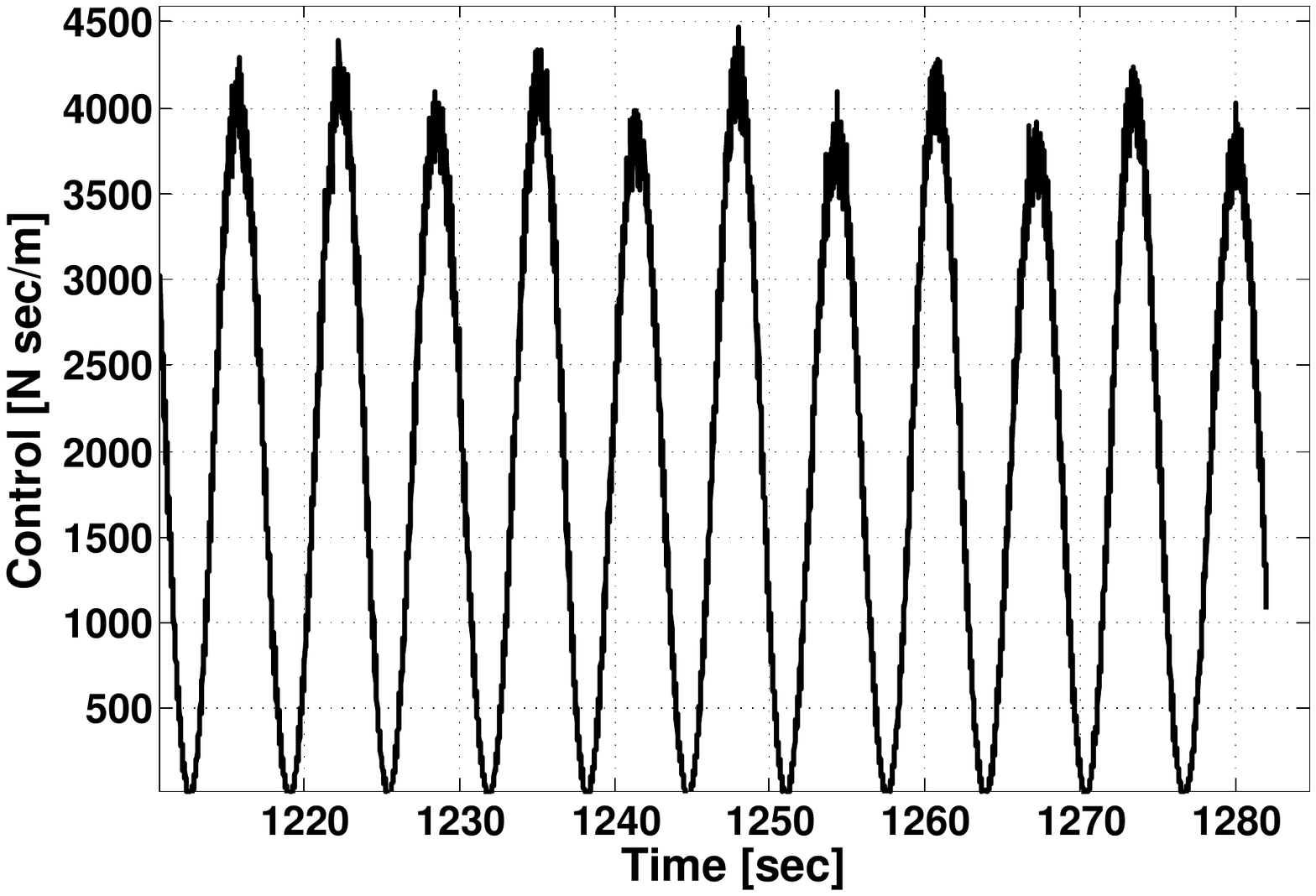}\vspace{-3cm}
\label{theorem3-sway-control-u-zoom} \caption{Output of controller
(\ref{control_1_t})- Zoom} \label{figure2_zoom}
\end{figure}

Next, we consider the model (\ref{ODE_model}),
(\ref{ODE_model_coefficients}) with no-zero disturbance signals:
$f_{1}(t)=0.2 sin(2\pi.0.08t)$, and $f_{2}$ being deduced from
$f_{1}$ via equation (\ref{BC_relation}). We underline that we
have purposely selected the disturbance frequency to be equal to
the first resonance frequency of the rope, to simulate the
`worst-case scenario'. In this case we apply the controller
(\ref{control_2_t}), with the parameters $u_{max_p}=10^9 \;N
sec/m,\;v_{1max}=v_{2max}=10^5\;N
sec/m,\;F_{max}=\tilde{F}_{max}=1$. We show on Figures
\ref{figure_1_with_dist}, \ref{figure_1_with_dist_zoom} the sway
signal in the uncontrolled and the controlled case. We see that
the sway steady state maximum amplitude is reduced from $8.4\;m$
in the uncontrolled case to $2.4\;m$ with control. The noisy
 (due to the simulated measurements noise) bounded and continuous control signals are reported on Figures
\ref{figure_2_with_dist}, \ref{figure_2_with_dist_zoom}.

\begin{figure}
\centering
\includegraphics[width=1\linewidth]{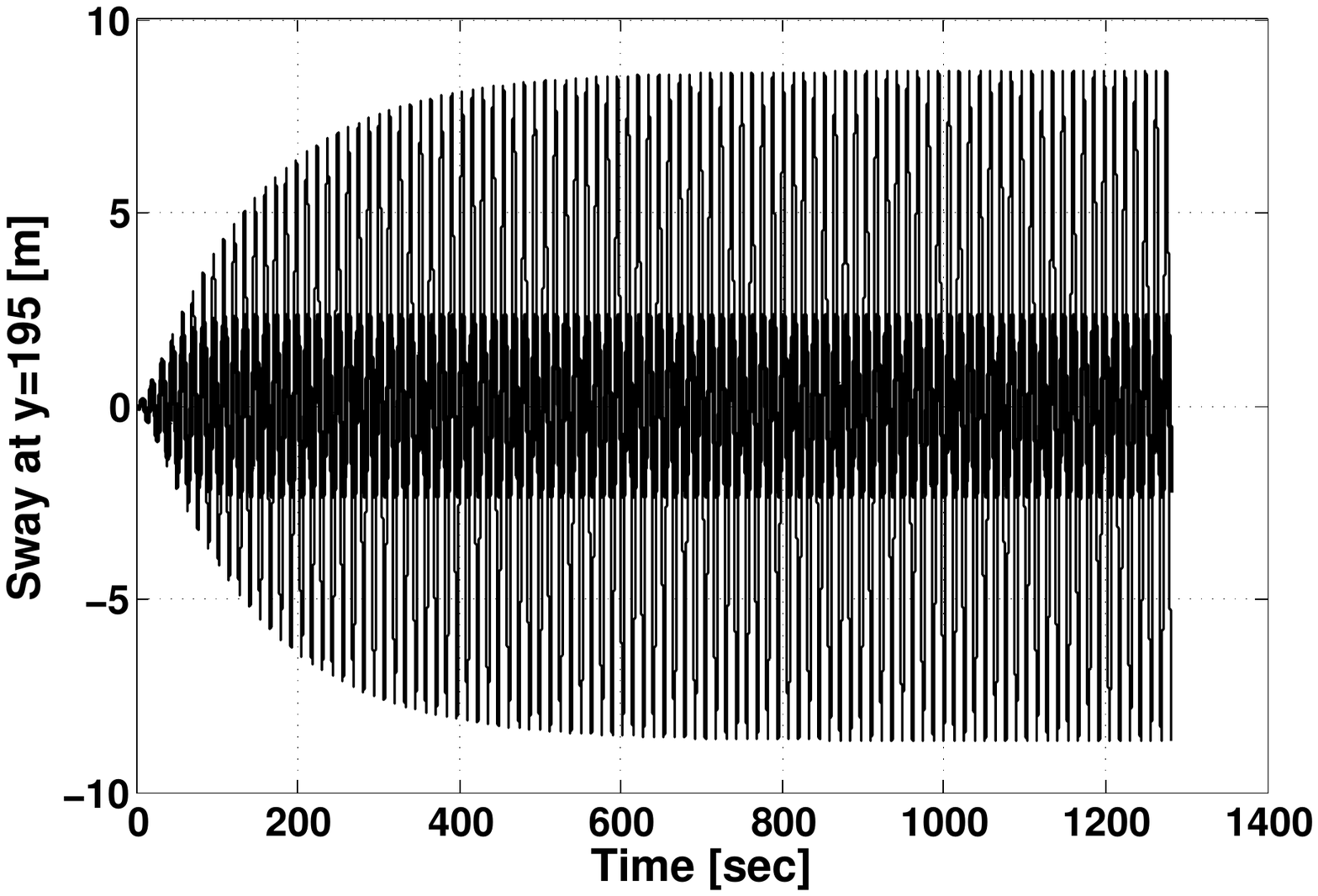}\vspace{-3cm}
\caption{Rope sway at $y=195\;m$: No control (thin line)- With
controller (\ref{control_2_t}) (bold line)}
\label{figure_1_with_dist} \vspace{-3cm}
\includegraphics[width=1\linewidth]{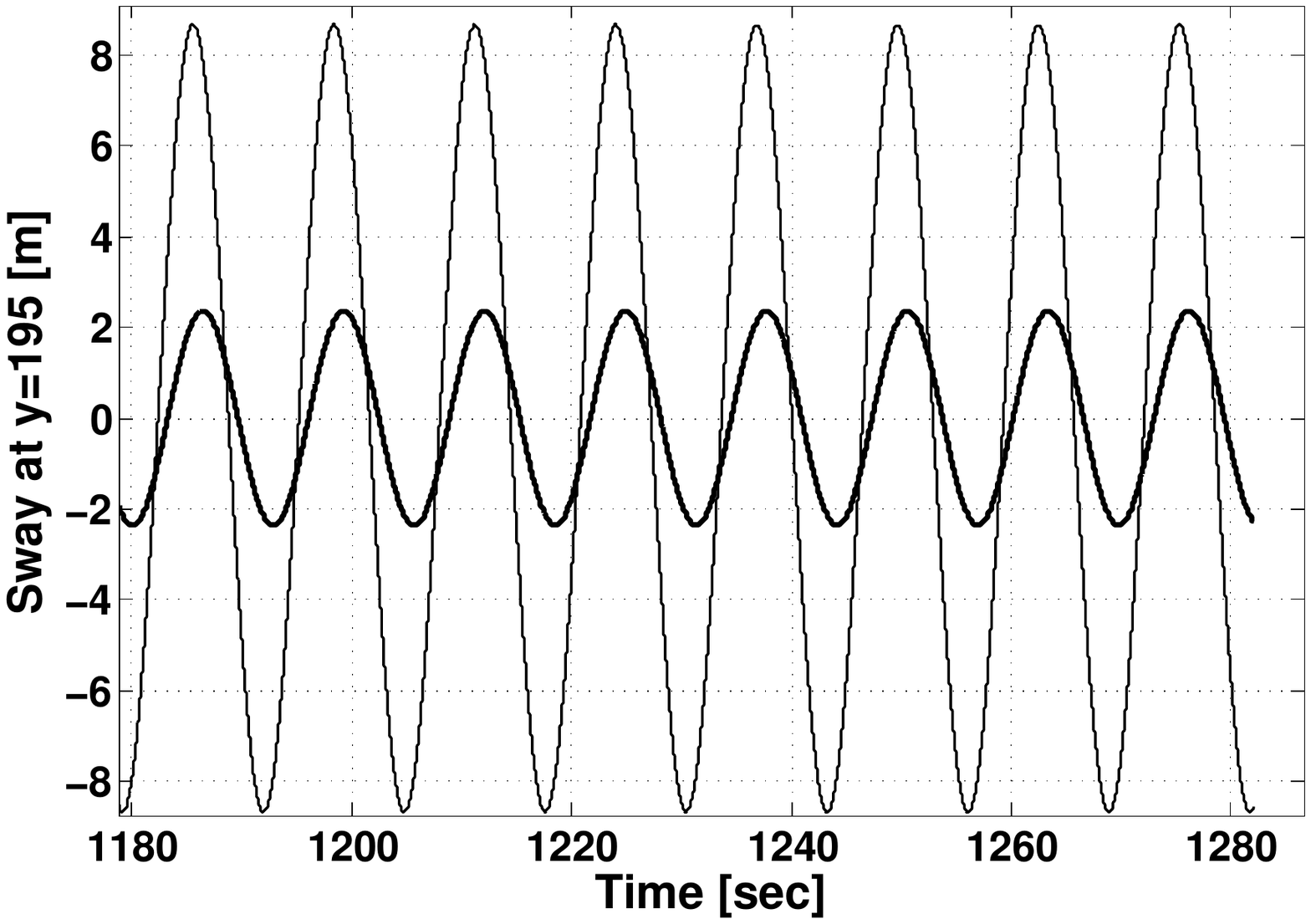}\vspace{-3cm}
\caption{Zoom of rope sway at $y=195\;m$: No control (thin line)-
With controller (\ref{control_2_t}) (bold line)}
\label{figure_1_with_dist_zoom}
\end{figure}
\begin{figure}
\centering
\includegraphics[width=1\linewidth]{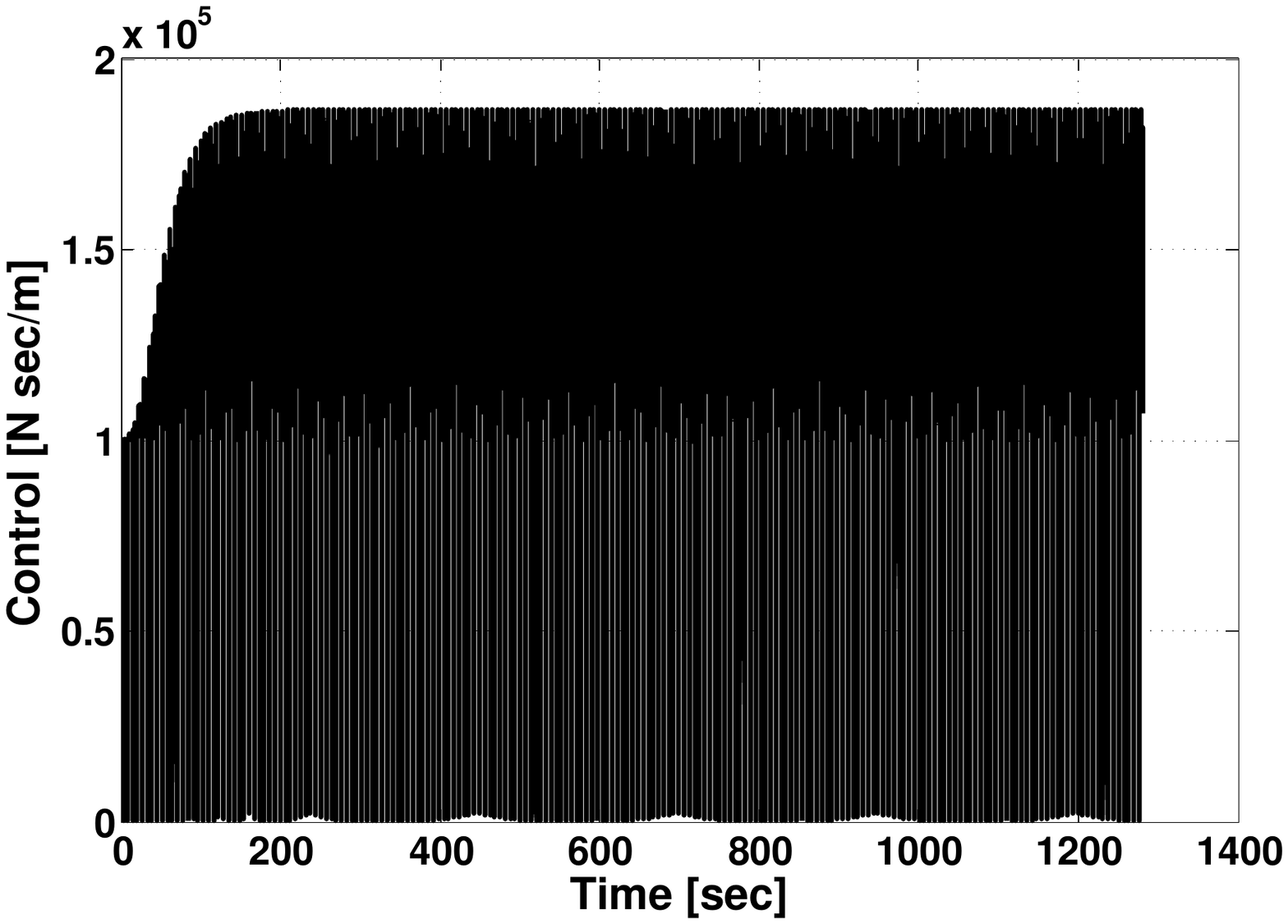}\vspace{-3cm}
\caption{Output of controller (\ref{control_2_t})}
\label{figure_2_with_dist} \vspace{-3cm}
\includegraphics[width=1\linewidth]{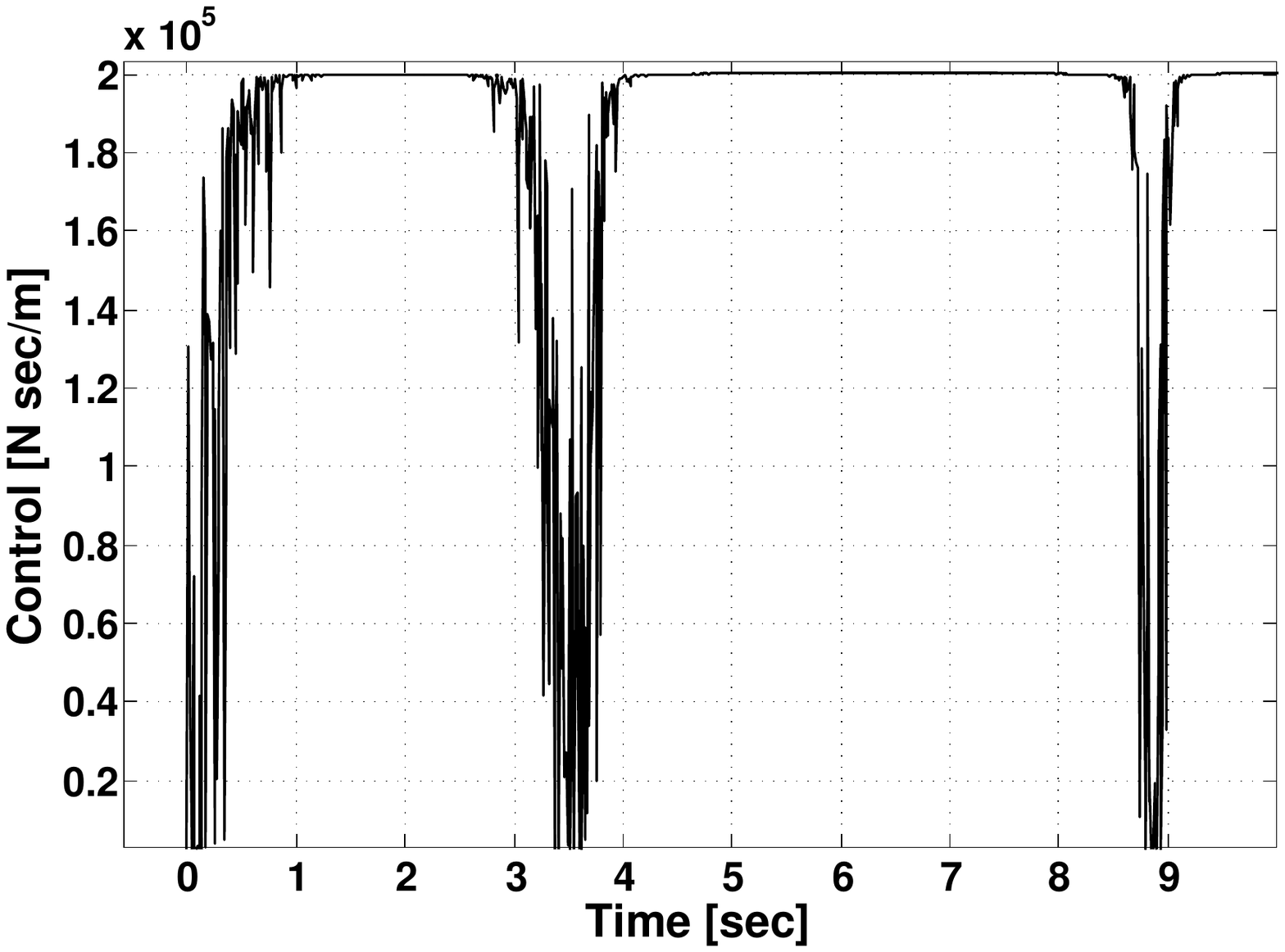}\vspace{-3cm}
 \caption{Output of controller (\ref{control_2_t})- Zoom}
 \label{figure_2_with_dist_zoom}
\end{figure}
\section{Conclusion}\label{section4}
In this paper we have studied the problem of semi active control
of elevator rope sway dynamics occurring due to external force
disturbances acting on the elevator system. We have considered the
case of a static car, i.e. constant rope length and have proposed
two nonlinear controllers based on Lyapunov theory. We have
presented the stability analysis of these controllers and shown
their efficiency on a numerical example. The semi-active
stabilization problems related to time-varying rope lengths, i.e.
moving car, will be presented in a future report.

\end{document}